\documentclass[letterpaper,11pt]{article}

\usepackage{amsthm}
\usepackage{amsfonts}
\usepackage{amsmath}
\usepackage{graphicx}
\usepackage{epstopdf}
\usepackage{lineno}
\usepackage{setspace} 
\usepackage{verbatim}
\usepackage{authblk}

\newtheorem{theorem}{Theorem}
\newtheorem*{theorem*}{Theorem}
\newtheorem{conjecture}[theorem]{Conjecture}

\newtheorem{lemma}[theorem]{Lemma}

\newcommand{\fc}{\mathcal C}

\newenvironment{enumerate*}{
\begin{enumerate}
  \setlength{\itemsep}{5pt}
  \setlength{\parskip}{0pt}
  \setlength{\parsep}{0pt}
}{\end{enumerate}}

\newenvironment{itemize*}{
\begin{itemize}
  \setlength{\itemsep}{5pt}
  \setlength{\parskip}{0pt}
  \setlength{\parsep}{0pt}
}{\end{itemize}}

\title{A local strengthening of Reed's $\omega$, $\Delta$, $\chi$ conjecture for quasi-line graphs}
\author{Maria Chudnovsky\thanks{Supported by NSF grants DMS-1001091 and IIS-1117631.}}\author{Andrew D.~King\thanks{Corresponding author.  Email: andrew.d.king@gmail.com.  Supported by an NSERC Postdoctoral Fellowship.}}\author{Matthieu Plumettaz\thanks{Email: mp2761@columbia.edu. Partially supported by NSF grant DMS-1001091.}}
\affil{Department of Industrial Engineering and Operations Research\\Columbia University, New York NY}
\author{Paul Seymour\thanks{Supported by ONR grant N00014-10-1-0680 and NSF grant DMS-0901075.}}
\affil{Department of Mathematics\\Princeton University, Princeton NJ}

\begin{document}

\maketitle

\begin{abstract}
Reed's $\omega$, $\Delta$, $\chi$ conjecture proposes that every graph satisfies $\chi\leq \lceil\frac 12(\Delta+1+\omega)\rceil$; it is known to hold for all claw-free graphs.  In this paper we consider a local strengthening of this conjecture.  We prove the local strengthening for line graphs, then note that previous results immediately tell us that the local strengthening holds for all quasi-line graphs.  Our proofs lead to polytime algorithms for constructing colourings that achieve our bounds: $O(n^2)$ for line graphs and $O(n^3m^2)$ for quasi-line graphs.  For line graphs, this is faster than the best known algorithm for constructing a colouring that achieves the bound of Reed's original conjecture.
\end{abstract}

\newpage

\section{Introduction}

All graphs and multigraphs we consider in this paper are finite.  Loops are permitted in multigraphs but not graphs.  Given a graph $G$ with maximum degree $\Delta(G)$ and clique number $\omega(G)$, the chromatic number $\chi(G)$ is trivially bounded above by $\Delta(G)+1$ and below by $\omega(G)$.  Reed's $\omega$, $\Delta$, $\chi$ conjecture proposes, roughly speaking, that $\chi(G)$ falls in the lower half of this range:

\begin{conjecture}[Reed]\label{con:reed}
For any graph $G$,
$$\chi(G)\leq \left\lceil \tfrac 12(\Delta(G)+1+\omega(G))\right\rceil.$$
\end{conjecture}

One of the first classes of graphs for which this conjecture was proved is the class of line graphs \cite{kingrv07}.  Already for line graphs the conjecture is tight, as evidenced by the strong product of $C_5$ and $K_\ell$ for any positive $\ell$; this is the line graph of the multigraph constructed by replacing each edge of $C_5$ by $\ell$ parallel edges.  The proof of Conjecture \ref{con:reed} for line graphs was later extended to quasi-line graphs \cite{kingthesis, kingr08} and later claw-free graphs \cite{kingthesis}  (we will define these graph classes shortly).  In his thesis, King proposed a local strengthening of Reed's conjecture.  For a vertex $v$, let $\omega(v)$ denote the size of the largest clique containing $v$.
\begin{conjecture}[King \cite{kingthesis}]\label{con:local}
For any graph $G$,
$$\chi(G) \leq   \max_{v\in V(G)}\left\lceil \tfrac 12(d(v)+1+\omega(v))\right\rceil.$$
\end{conjecture}

There are several pieces of evidence that lend credence to Conjecture \ref{con:local}.  First is the fact that the fractional relaxation holds.  This was noted by McDiarmid as an extension of a theorem of Reed \cite{molloyrbook}; the proof appears explicitly in \cite{kingthesis} \S2.2:

\begin{theorem}[McDiarmid]\label{thm:frac}
For any graph $G$,
$$\chi_f(G) \leq   \max_{v\in V(G)}\left( \tfrac 12(d(v)+1+\omega(v))\right).$$
\end{theorem}
The second piece of evidence for Conjecture \ref{con:local} is that the result holds for claw-free graphs with stability number at most three \cite{kingthesis}.  However, for the remaining classes of claw-free graphs, which are constructed as a generalization of line graphs \cite{cssurvey}, the conjecture has remained open.  In this paper we prove that Conjecture \ref{con:local} holds for line graphs.  We then show that we can extend this result to quasi-line graphs in the same way that Conjecture \ref{con:reed} was extended from line graphs to quasi-line graphs in \cite{kingr08}.  Our main result is:

\begin{theorem}\label{thm:ql}
For any quasi-line graph $G$,
$$\chi(G) \leq   \max_{v\in V(G)}\left\lceil \tfrac 12(d(v)+1+\omega(v))\right\rceil.$$
\end{theorem}

Furthermore our proofs yield polytime algorithms for constructing a proper colouring achieving the bound of the theorem: $O(n^2)$ time for a line graph on $n$ vertices, and $O(n^3m^2)$ time for a quasi-line graph on $n$ vertices and $m$ edges.

Given a multigraph $G$, the {\em line graph of $G$}, denoted $L(G)$, is the graph with vertex set $V(L(G))=E(G)$ in which two vertices of $L(G)$ are adjacent precisely if their corresponding edges in $H$ share an endpoint.  We say that a graph $G'$ is a {\em line graph} if for some multigraph $G$, $L(G)$ is isomorphic to $G'$.  A graph $G$ is {\em quasi-line} if every vertex $v$ is {\em bisimplicial}, i.e.\ the neighbourhood of $v$ induces the complement of a bipartite graph.  A graph $G$ is {\em claw-free} if it contains no induced $K_{1,3}$.  Observe that every line graph is quasi-line and every quasi-line graph is claw-free.

\section{Proving the local strengthening for line graphs}

In order to prove Conjecture \ref{con:local} for line graphs, we prove an equivalent statement in the setting of edge colourings of multigraphs.  Given distinct adjacent vertices $u$ and $v$ in a multigraph $G$, we let $\mu_G(uv)$ denote the number of edges between $u$ and $v$.  We let $t_G(uv)$ denote the maximum, over all vertices $w\notin \{u,v\}$, of the number of edges with both endpoints in $\{u,v,w\}$.  That is,
$$t_G(uv) := \max_{w\in N(u)\cap N(v)}\left(\mu_G(uv)+\mu_G(uw)+\mu_G(vw)\right).$$
We omit the subscripts when the multigraph in question is clear.

Observe that given an edge $e$ in $G$ with endpoints $u$ and $v$, the degree of $uv$ in $L(G)$ is $d(u)+d(v)-\mu(uv)-1$.  And since any clique in $L(G)$ containing $e$ comes from the edges incident to $u$, the edges incident to $v$, or the edges in a triangle containing $u$ and $v$, we can see that $\omega(v)$ in $L(G)$ is equal to $\max \{ d(u), d(v), t(uv) \}$.  Therefore we prove the following theorem, which, aside from the algorithmic claim, is equivalent to proving Conjecture \ref{con:local} for line graphs:

\begin{theorem}\label{thm:main}
Let $G$ be a multigraph on $m$ edges, and let
\begin{multline}\label{eq:main}
\gamma_l'(G) := \max_{uv\in E(G)} \left\lceil  \max \left\{ d(u)+\tfrac 12(d(v)-\mu(vu)),\ d(v)+\tfrac 12(d(u)-\mu(uv)),\right.\right. \\  \left.\left. \tfrac 12( d(u)+d(v)-\mu_G(uv)+t(uv)   )  \right\}\right\rceil.
\end{multline}
Then $\chi'(G)\leq \gamma_l'(G)$, and we can find a $\gamma_l'(G)$-edge-colouring of $G$ in $O(m^2)$ time. 
\end{theorem}

The most intuitive approach to achieving this bound on the chromatic index involves assuming that $G$ is a minimum counterexample, then characterizing $\gamma_l'(G)$-edge-colourings of $G-e$ for an edge $e$.  We want an algorithmic result, so we will have to be a bit more careful to ensure that we can modify partial $\gamma_l'(G)$-edge-colourings efficiently until we find one that we can extend to a complete $\gamma_l'(G)$-edge-colouring of $G$.

We begin by defining, for a vertex $v$, a {\em fan hinged at $v$}.  Let $e$ be an edge incident to $v$, and let $v_1,\ldots, v_\ell$ be a set of distinct neighbours of $v$ with $e$ between $v$ and $v_1$.  Let $c:E\setminus \{e\} \rightarrow \{1,\ldots,k\}$ be a proper edge colouring of $G\setminus \{e\}$ for some fixed $k$.  Then $F = (e;c;v;v_1,\ldots, v_\ell)$ is a {\em fan} if for every $j$ such that $2\leq j \leq \ell$, there exists some $i$ less than $j$ such that some edge between $v$ and $v_j$ is assigned a colour that does not appear on any edge incident to $v_i$ (i.e.\ a colour {\em missing} at $v_i$).  We say that $F$ is {\em hinged at $v$}.  If there is no $u \notin \{ v,v_1,\ldots,v_\ell \}$ such that $F'=(e;c;v;v_1,\ldots,v_\ell,u)$ is a fan, we say that $F$ is a {\em maximal fan}.  The {\em size} of a fan refers to the number of neighbours of the hinge vertex contained in the fan (in this case, $\ell$).  These fans generalize Vizing's fans, originally used in the proof of Vizing's theorem \cite{vizing64}.  Given a partial $k$-edge-colouring of $G$ and a vertex $w$, we say that a colour is {\em incident to $w$} if the colour appears on an edge incident to $w$.  We use $\fc(w)$ to denote the set of colours incident to $w$, and we use $\bar\fc(w)$ to denote $[k] \setminus \fc(w)$.

Fans allow us to modify partial $k$-edge-colourings of a graph (specifically those with exactly one uncoloured edge).  We will show that if $k\geq \gamma_l'(G)$, then either every maximal fan has size $2$ or we can easily find a $k$-edge-colouring of $G$.  We first prove that we can construct a $k$-edge-colouring of $G$ from a partial $k$-edge-colouring of $G-e$ whenever we have a fan for which certain sets are not disjoint.

\begin{lemma}\label{lem:algo1}
For some edge $e$ in a multigraph $G$ and positive integer $k$, let $c$ be a $k$-edge-colouring of $G-e$.  If there is a fan $F = (e;c;v;v_1,\ldots,v_\ell)$ such that for some $j$, $\bar \fc(v)\cap \bar \fc(v_j)\neq \emptyset$, then we can find a $k$-edge-colouring of $G$ in $O(k+m)$ time.
\end{lemma}

\begin{proof}
Let $j$ be the minimum index for which $\bar\fc(v)\cap \bar\fc(v_j)$ is nonempty.  If $j=1$ then the result is trivial, since we can extend $c$ to a proper $k$-edge-colouring of $G$.  Otherwise $j\geq 2$ and we can find $j$ in $O(m)$ time.  We define $e_1$ to be $e$.  We then construct a function $f:\{2,\ldots,\ell\} \rightarrow \{1,\ldots\ell-1\}$ such that for each $i$, (1) $f(i)<i$ and (2) there is an edge $e_i$ between $v$ and $v_i$ such that $c(e_i)$ is missing at $v_{f(i)}$.  We can find this function in $O(k+m)$ time by building a list of the earliest $v_i$ at which each colour is missing, and computing $f$ for increasing values of $i$ starting at 2.  While doing so we also find the set of edges $\{e_i\}_{i=2}^\ell$.

We construct a $k$-edge-colouring $c_j$ of $G-e_j$ from $c$ by shifting the colour $c(e_j)$ from $e_j$ to $e_{f(j)}$, shifting the colour $c(e_{f(j)})$ from $e_{f(j)}$ to $e_{f(f(j))}$, and so on, until we shift a colour to $e$.  We now have a $k$-edge-colouring $c_j$ of $G-e_j$ such that some colour is missing at both $v$ and $v_j$.  We can therefore extend $c_j$ to a proper $k$-edge-colouring of $G$ in $O(k+m)$ time.
\end{proof}

\begin{lemma}\label{lem:algo2}
For some edge $e$ in a multigraph $G$ and positive integer $k$, let $c$ be a $k$-edge-colouring of $G-e$.  If there is a fan $F = (e;c;v;v_1,\ldots,v_\ell)$ such that for some $i$ and $j$ satisfying $1\leq i<j\leq \ell$, $\bar \fc(v_i)\cap \bar \fc(v_j)\neq \emptyset$, then we can find $v_i$ and $v_j$ in $O(k+m)$ time, and we can find a $k$-edge-colouring of $G$ in $O(k+m)$ time.
\end{lemma}

\begin{proof}
We can easily find $i$ and $j$ in $O(k+m)$ time if they exist.  Let $\alpha$ be a colour in $\bar \fc(v)$ and let $\beta$ be a colour in $\bar \fc(v_i) \cap \bar \fc(v_j)$.  Note that by Lemma \ref{lem:algo1}, we can assume $\alpha \in  \fc(v_i) \cap \fc(v_j)$ and $\beta\in \fc(v)$.

Let $G_{\alpha,\beta}$ be the subgraph of $G$ containing those edges coloured $\alpha$ or $\beta$.  Every component of $G_{\alpha,\beta}$ containing $v$, $v_i$, or $v_j$ is a path on $\geq 2$ vertices.  Thus either $v_i$ or $v_j$ is in a component of $G_{\alpha,\beta}$ not containing $v$.  Exchanging the colours $\alpha$ and $\beta$ on this component leaves us with a $k$-edge-colouring of $G-e$ in which either $\bar \fc(v)\cap \bar \fc(v_i)\neq \emptyset$ or $\bar \fc(v)\cap \bar \fc(v_j)\neq \emptyset$.  This allows us to apply Lemma \ref{lem:algo1} to find a $k$-edge-colouring of $G$.  We can easily do this work in $O(m)$ time.
\end{proof}

The previous two lemmas suggest that we can extend a colouring more easily when we have a large fan, so we now consider how we can extend a fan that is not maximal.  Given a fan $F = (e;c;v;v_1,\ldots,v_\ell)$, we use $d(F)$ to denote $d(v)+\sum_{i=1}^\ell d(v_i)$.

\begin{lemma}\label{lem:extend}
For some edge $e$ in a multigraph $G$ and integer $k\geq \Delta(G)$, let $c$ be a $k$-edge-colouring of $G-e$ and let $F$ be a fan.  Then we can extend $F$ to a maximal fan $F' = (e;c;v;v_1,v_2,\ldots,v_\ell)$ in $O(k+d(F'))$ time.
\end{lemma}

\begin{proof}
We proceed by setting $F'=F$ and extending $F'$ until it is maximal.  To this end we maintain two colour sets.  The first, $\fc$, consists of those colours appearing incident to $v$ but not between $v$ and another vertex of $F'$.  The second, $\bar\fc_{F'}$, consists of those colours that are in $\fc$ and are missing at some fan vertex.  Clearly $F'$ is maximal if and only if $\bar\fc_{F'}=\emptyset$.  We can perform this initialization in $O(k+d(F))$ time by counting the number of times each colour in $\fc$ appears incident to a vertex of the fan.

Now suppose we have $F'=(e;c;v;v_1,v_2,\ldots,v_\ell)$, along with sets $\fc$ and $\bar\fc_{F'}$, which we may assume is not empty.  Take an edge incident to $v$ with a colour in $\bar\fc_F$; call its other endpoint $v_{\ell+1}$.  We now update $\fc$ by removing all colours appearing between $v$ and $v_{\ell+1}$.  We update $\bar\fc_{F'}$ by removing all colours appearing between $v$ and $v_{\ell+1}$, and adding all colours in $\fc \cap \bar\fc(v_{\ell+1})$.  Set $F'=(e;c;v;v_1,v_2,\ldots,v_{\ell+1})$.  We can perform this update in $d(v_{\ell+1})$ time; the lemma follows.
\end{proof}

We can now prove that if $k\geq \gamma_l'(G)$ and we have a maximal fan of size $1$ or at least $3$, we can find a $k$-edge-colouring of $G$ in $O(k+m)$ time.

\begin{lemma}\label{lem:single}
For some edge $e$ in a multigraph $G$ and positive integer $k \geq \gamma_l'(G)$, let $c$ be a $k$-edge-colouring of $G-e$ and let $F = (e;c;v;v_1)$ be a fan.  If $F$ is a maximal fan we can find a $k$-edge-colouring of $G$ in $O(k+m)$ time.
\end{lemma}
\begin{proof}
If $\bar\fc(v)\cap \bar\fc(v_1)$ is nonempty, then we can easily extend the colouring of $G-e$ to a $k$-edge-colouring of $G$.  So assume $\bar\fc(v) \cap \bar\fc(v_1)$ is empty.  Since $k\geq \gamma_l'(G)\geq 1$, $\bar\fc(v_1)$ is nonempty.  Therefore there is a colour in $\bar\fc(v_1)$ appearing on an edge incident to $v$ whose other endpoint, call it $v_2$, is not $v_1$.  Thus $(e;c;v;v_1,v_2)$ is a fan, contradicting the maximality of $F$.
\end{proof}

\begin{lemma}\label{lem:algo3}
For some edge $e$ in a multigraph $G$ and positive integer $k \geq \gamma_l'(G)$, let $c$ be a $k$-edge-colouring of $G-e$ and let $F=(e;c;v;v_1,v_2,\ldots,v_\ell)$ be a maximal fan with $\ell\geq 3$.  Then we can find a $k$-edge-colouring of $G$ in $O(k+m)$ time.
\end{lemma}

\begin{proof}
Let $v_0$ denote $v$ for ease of notation.  If the sets $\bar\fc(v_0), \bar\fc(v_1), \ldots, \bar\fc(v_\ell)$ are not all pairwise disjoint, then using Lemma \ref{lem:algo1} or Lemma \ref{lem:algo2} we can find a $k$-edge-colouring of $G$ in $O(m)$ time.  We can easily determine whether or not these sets are pairwise disjoint in $O(k+m)$ time.  Now assume they are all pairwise disjoint; we will exhibit a contradiction, which is enough to prove the lemma.

The number of missing colours at $v_i$, i.e.\ $|\bar \fc(v_i)|$, is $k-d(v_i)$ if $2\leq i\leq \ell$, and $k-d(v_i)+1$ if $i\in \{0,1\}$.  Since $F$ is maximal, any edge with one endpoint $v_0$ and the other endpoint outside $\{v_0,\ldots,v_\ell\}$ must have a colour not appearing in $\cup_{i=0}^\ell\bar\fc(v_i)$.  Therefore
\begin{equation}   \left( \sum_{i=0}^\ell k-d(v_i)  \right) + 2 + \left( d(v_0) - \sum_{i=1}^\ell \mu(v_0v_i)\right)  \ \ \leq  \ \  k
.\end{equation}
Thus
\begin{equation} \ell k + 2  - \sum_{i=1}^\ell \mu(v_0v_i)  \ \  \leq \ \ \sum_{i=1}^\ell d(v_i)
.\end{equation}
But since $k\geq \gamma_l'(G)$, (\ref{eq:main}) tells us that for all $i \in [\ell]$,
\begin{equation}\label{eq:2}
d(v_i) + \tfrac 12(d(v_0) - \mu(v_0v_i)) \ \ \leq \ \ k
\end{equation}
Thus substituting for $k$ tells us
\begin{equation*}
\sum_{i=1}^\ell \frac{d(v_0)+2d(v_i) -\mu(v_0v_i)}{2} + 2  - \sum_{i=1}^\ell \mu(v_0v_i) \ \  \leq \ \ \sum_{i=1}^\ell d(v_i).
\end{equation*}
So
\begin{eqnarray*}
2+\tfrac 12 \ell d(v_0) - \tfrac 32 \sum_{i=1}^\ell \mu(v_0v_i)  & \leq & 0\\
2 + \tfrac 12 \ell d(v_0)& \leq&\tfrac 32 \sum_{i=1}^\ell \mu(v_0v_i)\\
\tfrac \ell 2  d(v_0)& < &\tfrac 32 d(v_0).
\end{eqnarray*}
This is a contradiction, since $\ell \geq 3$.
\end{proof}

We are now ready to prove the main lemma of this section.

\begin{lemma}\label{lem:main}
For some edge $e_0$ in a multigraph $G$ and positive integer $k \geq \gamma_l'(G)$, let $c_0$ be a $k$-edge-colouring of $G-e$.  Then we can find a $k$-edge-colouring of $G$ in $O(k+m)$ time.
\end{lemma}

As we will show, this lemma easily implies Theorem \ref{thm:main}.  We approach this lemma by constructing a sequence of overlapping fans of size two until we can apply a previous lemma.  If we cannot do this, then our sequence results in a cycle in $G$ and a set of partial $k$-edge-colourings of $G$ with a very specific structure that leads us to a contradiction.

\begin{proof}
We postpone algorithmic considerations until the end of the proof.

Let $v_0$ and $v_1$ be the endpoints of $e_0$, and let $F_0=(e_0;c_0;v_1;v_0,u_1,\ldots,u_\ell)$ be a maximal fan.  If $|\{u_1,\ldots,u_\ell\}| \neq 1$ then we can apply Lemma \ref{lem:single} or Lemma \ref{lem:algo3}.  More generally, if at any time we find a fan of size three or more we can finish by applying Lemma \ref{lem:algo3}.  So assume $\{u_1,\ldots,u_\ell \}$ is a single vertex; call it $v_2$.

Let $\bar\fc_0$ denote the set of colours missing at $v_0$ in the partial colouring $c_0$, and take some colour $\alpha_0 \in \bar\fc_0$.  Note that if $\alpha_0$ does not appear on an edge between $v_1$ and $v_2$ then we can find a fan $(e_0;c_0;v_1;v_0,v_2,u)$ of size $3$ and apply Lemma \ref{lem:algo3} to complete the colouring.  So we can assume that $\alpha_0$ does appear on an edge between $v_1$ and $v_2$.

Let $e_1$ denote the edge between $v_1$ and $v_2$ given colour $\alpha_0$ in $c_0$.  We construct a new colouring $c_1$ of $G-e_1$ from $c_0$ by uncolouring $e_1$ and assigning $e_0$ colour $\alpha_0$.  Let $\bar \fc_1$ denote the set of colours missing at $v_1$ in the colouring $c_1$.  Now  let $F_1=(e_1;c_1;v_2;v_1,v_3)$ be a maximal fan.  As with $F_0$, we can assume that $F_1$ exists and is indeed maximal.  The vertex $v_3$ may or may not be the same as $v_0$.

Let $\alpha_1\in\bar\fc_1$ be a colour in $\bar\fc_1$.  Just as $\alpha_0$ appears between $v_1$ and $v_2$ in $c_0$, we can see that $\alpha_1$ appears between $v_2$ and $v_3$.  Now let $e_2$ be the edge between $v_2$ and $v_3$ having colour $\alpha_1$ in $c_1$.  We construct a colouring $c_2$ of $G-e_2$ from $c_1$ by uncolouring $e_2$ and assigning $e_1$ colour $\alpha_1$.  

We continue to construct a sequence of fans $F_i = (e_i,c_i;v_{i+1};v_i,v_{i+2})$ for $i=0,1,2,\ldots$ in this way, maintaining the property that $\alpha_{i+2}=\alpha_i$.  This is possible because when we construct $c_{i+1}$ from $c_i$, we make $\alpha_i$ available at $v_{i+2}$, so the set $\bar\fc_{i+2}$ (the set of colours missing at $v_{i+2}$ in the colouring $c_{i+2}$) always contains $\alpha_{i}$.  We continue constructing our sequence of fans until we reach some $j$ for which $v_j \in \{v_i\}_{i=0}^{j-1}$, which will inevitably happen if we never find a fan of size 3 or greater.  We claim that $v_j=v_0$ and $j$ is odd.  To see this, consider the original edge-colouring of $G-e_0$ and note that for $1\leq i\leq j-1$, $\alpha_{0}$ appears on an edge between $v_i$ and $v_{i+1}$ precisely if $i$ is odd, and $\alpha_1$ appears on an edge between $v_i$ and $v_{i+1}$ precisely if $i$ is even.  Thus since the edges of colour $\alpha_0$ form a matching, and so do the edges of colour $\alpha_1$, we indeed have $v_j=v_0$ and $j$ odd.  Furthermore $F_0=F_j$.  Let $C$ denote the cycle $v_0,v_1,\ldots,v_{j-1}$.  In each colouring, $\alpha_0$ and $\alpha_1$ both appear $(j-1)/2$ times on $C$, in a near-perfect matching.  Let $H$ be the sub-multigraph of $G$ consisting of those edges between $v_i$ and $v_{i+1}$ for $0\leq j\leq j-1$ (with indices modulo $j$).  Let $A$ be the set of colours missing on at least one vertex of $C$, and let $H_A$ be the sub-multigraph of $H$ consisting of $e_0$ and those edges receiving a colour in $A$ in $c_0$ (and therefore in any $c_i$).

Suppose $j=3$.  If some colour is missing on two vertices of $C$ in $c_0$, $c_1$, or $c_2$, we can easily find a $k$-edge-colouring of $G$ since any two vertices of $C$ are the endpoints of $e_0$, $e_1$, or $e_2$.  We know that every colour in $\bar\fc_0$ appears between $v_1$ and $v_2$, and every colour in $\bar\fc_1$ appears between $v_2$ and $v_3=v_0$.  Therefore $|E(H_A)|=|A|+1$.  Our construction tells us that every colour in $\bar\fc_0$ appears between $v_1$ and $v_2$, and every colour in $\bar\fc_1$ appears between $v_2$ and $v_3=v_0$.  Therefore
\begin{eqnarray*}
2\gamma_l'(G) &\geq& d_G(v_0)+d_G(v_1) +t_G(v_0v_1) -\mu_G(v_0v_1)\\
&=& d_{H_A}(v_0)+d_{H_A}(v_1)+2(k-|A|) + t_G(v_0v_1)  - \mu_G(v_0v_1)\\
&\geq& d_{H_A}(v_0)+d_{H_A}(v_1)+2(k-|A|) + t_{H_A}(v_0v_1)  - \mu_{H_A}(v_0v_1)\\
&\geq& 2|E(H_A)| + 2(k-|A|)\\
&>& 2|A|+2(k-|A|) = 2k
\end{eqnarray*}
This is a contradiction since $k\geq \gamma_l'(G)$.  We can therefore assume that $j\geq 5$.

Let $\beta$ be a colour in $A \setminus \{\alpha_0,\alpha_1\}$.  If $\beta$ is missing at two consecutive vertices $v_i$ and $v_{i+1}$ then we can easily extend $c_i$ to a $k$-edge-colouring of $G$.  Bearing in mind that each $F_i$ is a maximal fan, we claim that if $\beta$ is not missing at two consecutive vertices then either we can easily $k$-edge-colour $G$, or the number of edges coloured $\beta$ in $H_{A}$ is at least twice the number of vertices at which $\beta$ is missing in any $c_i$.

To prove this claim, first assume without loss of generality that $\beta\in \bar\fc_0$.  Since $\beta$ is not missing at $v_1$, $\beta$ appears on an edge between $v_1$ and $v_2$ for the same reason that $\alpha_0$ does.  Likewise, since $\beta$ is not missing at $v_{j-1}$, $\beta$ appears on an edge between $v_{j-1}$ and $v_{j-2}$.  Finally, suppose $\beta$ appears between $v_1$ and $v_2$, and is missing at $v_3$ in $c_0$.  Then let $e_\beta$ be the edge between $v_1$ and $v_2$ with colour $\beta$ in $c_0$.  We construct a colouring $c'_0$ from $c_0$ by giving $e_2$ colour $\beta$ and giving $e_\beta$ colour $\alpha_1$ (i.e.\ we swap the colours of $e_\beta$ and $e_2$).  Thus $c'_0$ is a $k$-edge-colouring of $G-e_0$ in which $\beta$ is missing at both $v_0$ and $v_1$.  We can therefore extend $G-e_0$ to a $k$-edge-colouring of $G$.  Thus if $\beta$ is missing at $v_{3}$ or $v_{j-3}$ we can easily $k$-edge-colour $G$.  We therefore have at least two edges of $H_A$ coloured $\beta$ for every vertex of $C$ at which $\beta$ is missing, and we do not double-count edges.  This proves the claim, and the analogous claim for any colour in $A$ also holds.

Now we have
\begin{equation}
\sum_{i=0}^{j-1}\mu_{H_A}(v_iv_{i+1}) = |E(H_A)|  \ > \ 2 \sum_{i=0}^{j-1}\left(k - d_{G}(v_i) \right).
\end{equation}
Therefore taking indices modulo $j$, we have
\begin{equation}
\sum_{i=0}^{j-1}\left( d_{G}(v_i)+\tfrac 12\mu_{H_A}(v_{i+1}v_{i+2}) \right) \ > \ jk.
\end{equation}
Therefore there exists some index $i$ for which 
\begin{eqnarray}
d_{G}(v_i)+\tfrac 12\mu_{H_A}(v_{i+1}v_{i+2})  &> & k.
\end{eqnarray}
Therefore
\begin{eqnarray}
k\geq d_{G}(v_i)+\tfrac 12\mu_{G}(v_{i+1}v_{i+2})  &>& k.
\end{eqnarray}
This is a contradiction, so we can indeed find a $k$-edge-colouring of $G$.  It remains to prove that we can do so in $O(k+m)$ time.

Given the colouring $c_i$, we can construct the fan $F_i = (e_i,c_i;v_{i+1};v_i,v_{i+2})$ and determine whether or not it is maximal in $O(k+d(F_i))$ time.  If it is not maximal, we can complete the $k$-edge-colouring of $G$ in $O(m)$ time; this will happen at most once throughout the entire process.  Therefore we will either complete the colouring or construct our cycle of fans $F_0,\ldots,F_{j-1}$ in $O(\sum_{i=0}^{j-1}(k+d(F_i)))$ time.  This is not the desired bound, so suppose there is an index $i$ for which $k> d(F_i)$.  In this case we certainly have two intersecting sets of available colours in $F_i$, so we can apply Lemma \ref{lem:algo1} or \ref{lem:algo2} when we arrive at $F_i$, and find the $k$-edge-colouring of $G$ in $O(k+m)$ time.  If no such $i$ exists, then $jk = O(\sum_{i=0}^{j-1}(d(F_i))) = O(m)$, and we indeed complete the construction of all fans in $O(k+m)$ time.

Since each $F_i$ is a maximal fan, in $c_0$ there must be some colour $\beta\notin \{\alpha_0,\alpha_1\}$ missing at two consecutive vertices $v_i$ and $v_{i+1}$, otherwise we reach a contradiction.  We can find this $\beta$ and $i$ by going around the cycle of fans and comparing $\bar\fc_i$ and $\bar\fc_{i+1}$, and since this is trivial if $|\bar\fc_i|+|\bar\fc_{i+1}| > d(v_i)+d(v_{i+1})$ we can find $\beta$ and $i$ in $O(k+m)$ time, after which it is easy to construct the $k$-edge-colouring of $G$ from $c_i$. 
\end{proof}

We now complete the proof of Theorem \ref{thm:main}.

\begin{proof}[Proof of Theorem \ref{thm:main}]
Order the edges of $G$ $e_1,\ldots,e_m$ arbitrarily and let $k=\gamma_l'(G)$, which we can easily compute in $O(nm)$ time.  For $i=0,\ldots,m$, let $G_i$ denote the subgraph of $G$ on edges $\{e_j \mid j\leq i\}$.  Since $G_0$ is empty it is vacuously $k$-edge-coloured.  Given a $k$-edge-colouring of $G_i$, we can find a $k$-edge-colouring of $G_{i+1}$ in $O(k+m)$ time by applying Lemma \ref{lem:main}.  Since $k=\gamma_l'(G)=O(m)$, the theorem follows.
\end{proof}

This gives us the following result for line graphs, since for any multigraph $G$ we have $|V(L(G))| = |E(G)|$:

\begin{theorem}
Given a line graph $G$ on $n$ vertices, we can find a proper colouring of $G$ using $\gamma_l(G)$ colours in $O(n^2)$ time.
\end{theorem}

This is faster than the algorithm of King, Reed, and Vetta \cite{kingrv07} for $\gamma(G)$-colouring line graphs, which is given an improved complexity bound of $O(n^{5/2})$ in \cite{kingthesis}, \S4.2.3.

\section{Extending the result to quasi-line graphs}

We now leave the setting of edge colourings of multigraphs and consider vertex colourings of simple graphs.  As mentioned in the introduction, we can extend Conjecture \ref{con:local} from line graphs to quasi-line graphs using the same approach that King and Reed used to extend Conjecture \ref{con:reed} from line graphs to quasi-line graphs in \cite{kingr08}.  We do not require the full power of Chudnovsky and Seymour's structure theorem for quasi-line graphs \cite{clawfree7}.  Instead, we use a simpler decomposition theorem from \cite{cssurvey}.

\subsection{The structure of quasi-line graphs}

We wish to describe the structure of quasi-line graphs.  If a quasi-line graph does not contain a certain type of homogeneous pair of cliques, then it is either a circular interval graph or built as a generalization of a line graph -- where in a line graph we would replace each edge with a vertex, we now replace each edge with a linear interval graph.  We now describe this structure more formally.

\subsubsection{Linear and circular interval graphs}

A {\em linear interval graph} is a graph $G=(V,E)$ with a {\em linear interval representation}, which is a point on the real line for each vertex and a set of intervals, such that vertices $u$ and $v$ are adjacent in $G$ precisely if there is an interval containing both corresponding points on the real line.  If $X$ and $Y$ are specified cliques in $G$ consisting of the $|X|$ leftmost and $|Y|$ rightmost vertices (with respect to the real line) of $G$ respectively, we say that $X$ and $Y$ are {\em end-cliques} of $G$.  These cliques may be empty.

Accordingly, a {\em circular interval graph} is a graph with a {\em circular interval representation}, i.e.\ $|V|$ points on the unit circle and a set of intervals (arcs) on the unit circle such that two vertices of $G$ are adjacent precisely if some arc contains both corresponding points.  Circular interval graphs are the first of two fundamental types of quasi-line graph.  Deng, Hell, and Huang proved that we can identify and find a representation of a circular or linear interval graph in $O(m)$ time \cite{denghh96}.

\subsubsection{Compositions of linear interval strips}

We now describe the second fundamental type of quasi-line graph.

A {\em linear interval strip} $(S,X,Y)$ is a linear interval graph $S$ with specified end-cliques $X$ and $Y$.  We compose a set of strips as follows.  We begin with an underlying directed multigraph $H$, possibly with loops, and for every every edge $e$ of $H$ we take a linear interval strip $(S_e, X_e, Y_e)$.  For $v\in V(H)$ we define the {\em hub clique} $C_v$ as
$$C_v = \left( \bigcup\{X_e \mid e \textrm{ is an edge out of } v \}\right) \cup \left( \bigcup \{Y_e \mid e \textrm{ is an edge into } v \}\right).$$
We construct $G$ from the disjoint union of $\{ S_e \mid e \in E(H)\}$ by making each $C_v$ a clique; $G$ is then a {\em composition of linear interval strips} (see Figure \ref{fig:strip}).  Let $G_h$ denote the subgraph of $G$ induced on the union of all hub cliques.  That is,
$$G_h = G[\cup_{v\in V(H)} C_v] = G[\cup_{e\in E(H)} (X_e\cup Y_e)].$$

\begin{figure}
\begin{center}
\includegraphics[width=0.7\textwidth]{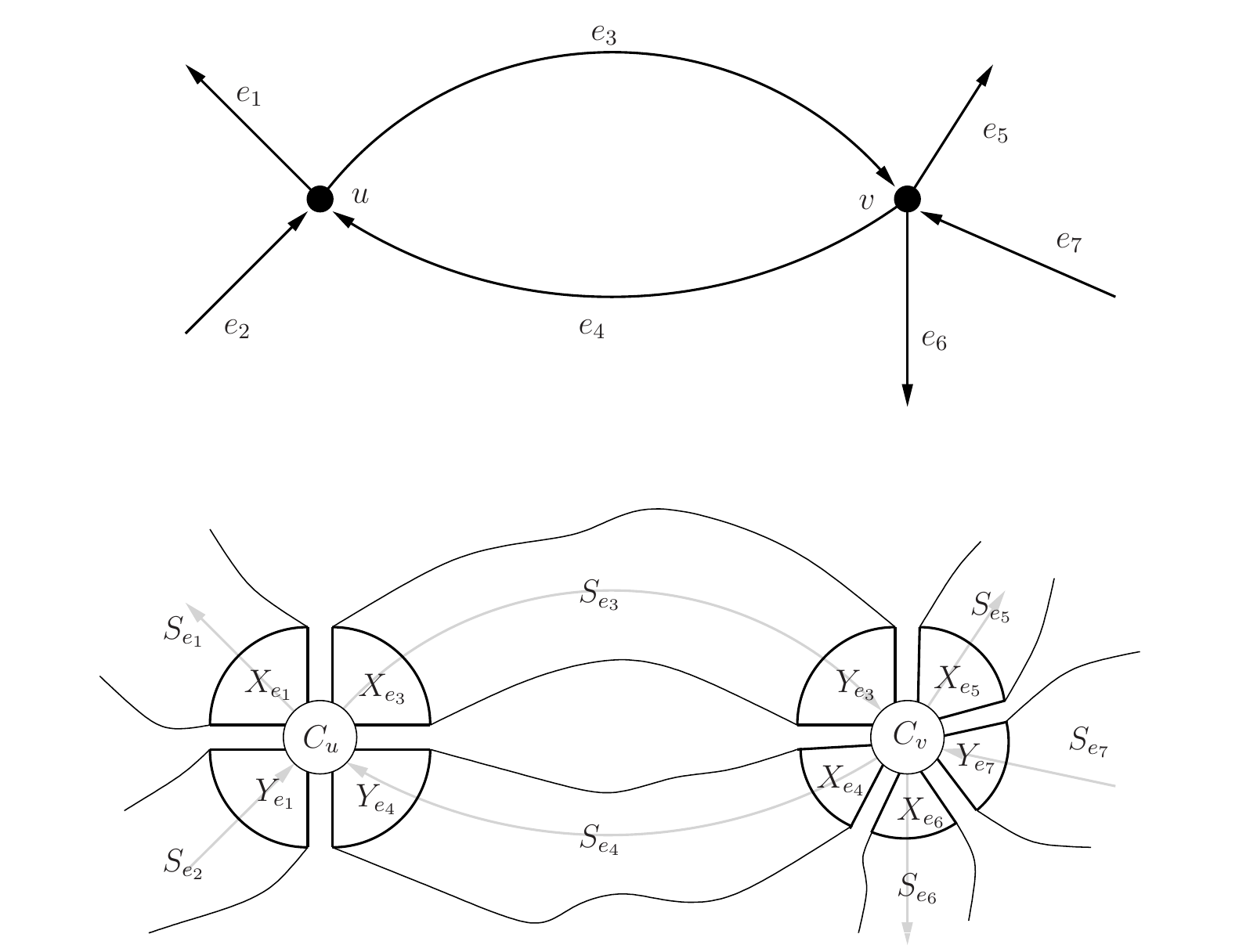}
\end{center}
\caption[A composition of strips]{We compose a set of strips $\{(S_e,X_e,Y_e) \mid e\in E(H)\}$ by joining them together on their end-cliques.  A hub clique $C_u$ will arise for each vertex $u \in V(H)$.}
\label{fig:strip}
\end{figure}

Compositions of linear interval strips generalize line graphs: note that if each $S_e$ satisfies $|S_e|=|X_e|=|Y_e|=1$ then $G = G_h = L(H)$.  

\subsubsection{Homogeneous pairs of cliques}

A pair of disjoint nonempty cliques $(A,B)$ in a graph is a {\em homogeneous pair of cliques} if $|A|+|B|\geq 3$, every vertex outside $A\cup B$ is adjacent to either all or none of $A$, and every vertex outside $A\cup B$ is adjacent to either all or none of $B$.  Furthermore $(A,B)$ is {\em nonlinear} if $G$ contains an induced $C_4$ in $A\cup B$ (this condition is equivalent to insisting that the subgraph of $G$ induced by $A\cup B$ is a linear interval graph).

\subsubsection{The structure theorem}

Chudnovsky and Seymour's structure theorem for quasi-line graphs \cite{cssurvey} tells us that any quasi-line graph not containing a clique cutset is made from the building blocks we just described.

\begin{theorem}
Any quasi-line graph containing no clique cutset and no nonlinear homogeneous pair of cliques is either a circular interval graph or a composition of linear interval strips.
\end{theorem}

To prove Theorem \ref{thm:ql}, we first explain how to deal with circular interval graphs and nonlinear homogeneous pairs of cliques, then move on to considering how to decompose a composition of linear interval strips.

\subsection{Circular interval graphs}

We can easily prove Conjecture \ref{con:local} for circular interval graphs by combining previously known results.  Niessen and Kind proved that every circular interval graph $G$ satisfies $\chi(G) = \lceil \chi_f(G)\rceil$ \cite{niessenk00}, so Theorem \ref{thm:frac} immediately implies that Conjecture \ref{con:local} holds for circular interval graphs.  Furthermore Shih and Hsu \cite{shihh89} proved that we can optimally colour circular interval graphs in $O(n^{3/2})$ time, which gives us the following result:

\begin{lemma}\label{lem:cig}
Given a circular interval graph $G$ on $n$ vertices, we can $\gamma_l(G)$-colour $G$ in $O(n^{3/2})$ time.
\end{lemma}

\subsection{Nonlinear homogeneous pairs of cliques}


There are many lemmas of varying generality that tell us we can easily deal with nonlinear homogeneous pairs of cliques; we use the version used by King and Reed \cite{kingr08} in their proof of Conjecture \ref{con:reed} for quasi-line graphs:

\begin{lemma}\label{lem:homo}
Let $G$ be a quasi-line graph on $n$ vertices containing a nonlinear homogeneous pair of cliques $(A,B)$.  In $O(n^{5/2})$ time we can find a proper subgraph $G'$ of $G$ such that $G'$ is quasi-line, $\chi(G')=\chi(G)$, and  given a $k$-colouring of $G'$ we can find a $k$-colouring of $G$ in $O(n^{5/2})$ time.
\end{lemma}

It follows immediately that no minimum counterexample to Theorem \ref{thm:ql} contains a nonlinear homogeneous pair of cliques.

\subsection{Decomposing: Clique cutsets}

Decomposing graphs on clique cutsets for the purpose of finding vertex colourings is straightforward and well understood.  

For any monotone bound on the chromatic number for a hereditary class of graphs, no minimum counterexample can contain a clique cutset, since we can simply ``paste together'' two partial colourings on a clique cutset.  Tarjan \cite{tarjan85} gave an $O(nm)$-time algorithm for constructing a clique cutset decomposition tree of any graph, and noted that given $k$-colourings of the leaves of this decomposition tree, we can construct a $k$-colouring of the original graph in $O(n^2)$ time.  Therefore if we can $\gamma_l(G)$-colour any quasi-line graph containing no clique cutset in $O(f(n,m))$ time for some function $f$, we can $\gamma_l(G)$-colour any quasi-line graph in $O(f(n,m)+nm)$ time.

If the multigraph $H$ contains a loop or a vertex of degree 1, then as long as $G$ is not a clique, it will contain a clique cutset.

\subsection{Decomposing: Canonical interval 2-joins}

A canonical interval 2-join is a composition by which a linear interval graph is attached to another graph.  Canonical interval 2-joins arise from compositions of strips, and can be viewed as a local decomposition rather than one that requires knowledge of a graph's global structure as a composition of strips.

Given four cliques $X_1$, $Y_1$, $X_2$, and $Y_2$, we say that $((V_1,X_1,Y_1),(V_2,X_2,Y_2))$ is  an {\em interval 2-join} if it satisfies the following:
\begin{itemize}
\item $V(G)$ can be partitioned into nonempty $V_1$ and $V_2$ with $X_1\cup Y_1\subseteq V_1$ and $X_2\cup Y_2\subseteq V_2$ such that for $v_1\in V_1$ and $v_2\in V_2$, $v_1v_2$ is an edge precisely if $\{v_1,v_2\}$ is in $X_1\cup X_2$ or $Y_1\cup Y_2$.
\item $G|V_2$ is a linear interval graph with end-cliques $X_2$ and $Y_2$.
\end{itemize}
If we also have $X_2$ and $Y_2$ disjoint, then we say $((V_1,X_1,Y_1),(V_2,X_2,Y_2))$ is a {\em canonical interval 2-join}.  The following decomposition theorem is a straightforward consequence of the structure theorem for quasi-line graphs:

\begin{theorem}\label{thm:decomposition}
Let $G$ be a quasi-line graph containing no nonlinear homogeneous pair of cliques.  Then one of the following holds.
\begin{itemize*}
\item $G$ is a line graph
\item $G$ is a circular interval graph
\item $G$ contains a clique cutset
\item $G$ admits a canonical interval 2-join.
\end{itemize*}
\end{theorem}

Therefore to prove Theorem \ref{thm:ql} it only remains to prove that a minimum counterexample cannot contain a canonical interval 2-join.  Before doing so we must give some notation and definitions.

We actually need to bound a refinement of $\gamma_l(G)$.  Given a canonical interval 2-join $((V_1,X_1, Y_1),(V_2,X_2, Y_2))$ in $G$ with an appropriate partitioning $V_1$ and $V_2$, let $G_1$ denote $G|V_1$, let $G_2$ denote $G|V_2$ and let $H_2$ denote $G|(V_2 \cup X_1 \cup Y_1)$.  For $v \in H_2$ we define $\omega'(v)$ as the size of the largest clique in $H_2$ containing $v$ and not intersecting both $X_1 \setminus Y_1$ and $Y_1 \setminus X_1$, and we define $\gamma_l^j(H_2)$ as $\max_{v\in H_2}\lceil d_G(v)+1+\omega'(v)\rceil$ (here the superscript $j$ denotes {\em join}).  Observe that $\gamma_l^j(H_2) \leq \gamma_l(G)$.  If $v \in X_1 \cup Y_1$, then $\omega'(v)$ is $|X_1|+|X_2|$, $|Y_1|+|Y_2|$, or $|X_1\cap Y_1|+ \omega(G|(X_2 \cup Y_2))$.

The following lemma is due to King and Reed and first appeared in \cite{kingthesis}; we include the proof for the sake of completeness.

\begin{lemma}\label{lem:quasilinemce2}
Let $G$ be a graph on $n$ vertices and suppose $G$ admits a canonical interval 2-join $((V_1,X_1, Y_1),(V_2,X_2, Y_2))$.  Then given a proper $l$-colouring of $G_1$ for any $l \geq \gamma_l^j(H_2)$, we can find a proper $l$-colouring of $G$ in $O(nm)$ time.
\end{lemma}

Since $\gamma_l^j(H_2) \leq \gamma_l(G)$, this lemma implies that no minimum counterexample to Theorem \ref{thm:ql} contains a canonical interval 2-join.

It is easy to see that a minimum counterexample cannot contain a simplicial vertex (i.e.\ a vertex whose neighbourhood is a clique).  Therefore in a canonical interval 2-join $((V_1,X_1,Y_1),(V_2,X_2,Y_2))$ in a minimum counterexample, all four cliques $X_1$, $Y_1$, $X_2$, and $Y_2$ must be nonempty.

\begin{proof}
We proceed by induction on $l$, observing that the case $l = 1$ is trivial.  We begin by modifying the colouring so that the number $k$ of colours used in both $X_1$ and $Y_1$ in the $l$-colouring of $G_1$ is maximal.  That is, if a vertex $v \in X_1$ gets a colour that is not seen in $Y_1$, then every colour appearing in $Y_1$ appears in $N(v)$.  This can be done in $O(n^2)$ time.  If $l$ exceeds $\gamma_l^j(H_2)$ we can just remove a colour class in $G_1$ and apply induction on what remains.  Thus we can assume that $l = \gamma_l^j(H_2)$ and so if we apply induction we must remove a stable set whose removal lowers both $l$ and $\gamma_l^j(H_2)$.

We use case analysis; when considering a case we may assume no previous case applies.  In some cases we extend the colouring of $G_1$ to an $l$-colouring of $G$ in one step.  In other cases we remove a colour class in $G_1$ together with vertices in $G_2$ such that everything we remove is a stable set, and when we remove it we reduce $\gamma_l^j(v)$ for every $v \in H_2$; after doing this we apply induction on $l$.  Notice that if $X_1 \cap Y_1 \neq \emptyset$ and there are edges between $X_2$ and $Y_2$ we may have a large clique in $H_2$ which contains some but not all of $X_1$ and some but not all of $Y_1$; this is not necessarily obvious but we deal with it in every applicable case.

\begin{itemize}
\item[Case 1.]$Y_1 \subseteq X_1$.

$H_2$ is a circular interval graph and $X_1$ is a clique cutset.  We can $\gamma_l(H_2)$-colour $H_2$ in $O(n^{3/2})$ time using Lemma \ref{lem:cig}.  By permuting the colour classes we can ensure that this colouring agrees with the colouring of $G_1$.  In this case $\gamma_l(H_2) \leq \gamma_l^j(H_2) \leq l$ so we are done.  By symmetry, this covers the case in which $X_1 \subseteq Y_1$.

\item[Case 2.]$k=0$ and $l > |X_1|+|Y_1|$.

Here $X_1$ and $Y_1$ are disjoint.  Take a stable set $S$ greedily from left to right in $G_2$.  By this we mean that we start with $S=\{v_1\}$, the leftmost vertex of $X_2$, and we move along the vertices of $G_2$ in linear order, adding a vertex to $S$ whenever doing so will leave $S$ a stable set.  So $S$ hits $X_2$.  If it hits $Y_2$, remove $S$ along with a colour class in $G_1$ not intersecting $X_1\cup Y_1$; these vertices together make a stable set.  If $v\in G_2$ it is easy to see that $\gamma_l^j(v)$ will drop: every remaining vertex in $G_2$ either loses two neighbours or is in $Y_2$, in which case $S$ intersects every maximal clique containing $v$.  If $v\in X_1 \cup Y_1$ then since $X_1$ and $Y_1$ are disjoint, $\omega'(v)$ is either $|X_1|+|X_2|$ or $|Y_1|+|Y_2|$; in either case $\omega'(v)$, and therefore $\gamma_l^j(v)$, drops when $S$ and the colour class are removed.  Therefore $\gamma_l^j(H_2)$ drops, and we can proceed by induction.

If $S$ does not hit $Y_2$ we remove $S$ along with a colour class from $G_1$ that hits $Y_1$ (and therefore not $X_1$).  Since $S\cap Y_2 = \emptyset$ the vertices together make a stable set.  Using the same argument as before we can see that removing these vertices drops both $l$ and $\gamma_l^j(H_2)$, so we can proceed by induction.

\item[Case 3.]$k=0$ and $l = |X_1|+|Y_1|$.

Again, $X_1$ and $Y_1$ are disjoint.  By maximality of $k$, every vertex in $X_1 \cup Y_1$ has at least $l-1$ neighbours in $G_1$.  Since $l=|X_1|+|Y_1|$ we know that $\omega'(X_1) \leq |X_1|+|Y_1|-|X_2|$ and $\omega'(Y_1) \leq |X_1|+|Y_1|-|Y_2|$.  Thus $|Y_1| \geq 2|X_2|$ and similarly $|X_1| \geq 2|Y_2|$.  Assume without loss of generality that $|Y_2| \leq |X_2|$.

We first attempt to $l$-colour $H_2 - Y_1$, which we denote by $H_3$, such that every colour in $Y_2$ appears in $X_1$ -- this is clearly sufficient to prove the lemma since we can permute the colour classes and paste this colouring onto the colouring of $G_1$ to get a proper $l$-colouring of $G$.  If $\omega(H_3) \leq l-|Y_2|$ then this is easy:  we can $\omega(H_3)$-colour the vertices of $H_3$, then use $|Y_2|$ new colours to recolour $Y_2$ and $|Y_2|$ vertices of $X_1$.  This is possible since $Y_2$ and $X_1$ have no edges between them.

Define $b$ as $l - \omega(H_3)$; we can assume that $b < |Y_2|$.  We want an $\omega(H_3)$-colouring of $H_3$ such that at most $b$ colours appear in $Y_2$ but not $X_1$.  There is some clique $C = \{v_i, \ldots, v_{i+\omega(H_3)-1}\}$ in $H_3$; this clique does not intersect $X_1$ because $|X_1 \cup X_2| \leq l-\frac 12|Y_1| \leq l - |Y_2|< l-b$.  Denote by $v_j$ the leftmost neighbour of $v_i$.  Since $\gamma_l^j(v_i) \leq l$, it is clear that $v_i$ has at most $2b$ neighbours outside $C$, and since $b < |Y_2| \leq \frac 12 |X_1|$ we can be assured that $v_i \notin X_2$.  Since $\omega(H_3)>|Y_2|$, $v_i \notin Y_2$.

We now colour $H_3$ from left to right, modulo $\omega(H_3)$.  If at most $b$ colours appear in $Y_2$ but not $X_1$ then we are done, otherwise we will ``roll back'' the colouring, starting at $v_i$.  That is, for every $p \geq i$, we modify the colouring of $H_3$ by giving $v_p$ the colour after the one that it currently has, modulo $\omega(H_3)$.  Since $v_i$ has at most $2b$ neighbours behind it, we can roll back the colouring at least $\omega(H_3)-2b-1$ times for a total of $\omega(H_3)-2b$ proper colourings of $H_3$.

Since $v_i \notin Y_2$ the colours on $Y_2$ will appear in order modulo $\omega(H_3)$.  Thus there are $\omega(H_3)$ possible sets of colours appearing on $Y_2$, and in $2b+1$ of them there are at most $b$ colours appearing in $Y_2$ but not $X_1$.  It follows that as we roll back the colouring of $H_3$ we will find an acceptable colouring.

Henceforth we will assume that $|X_1| \geq |Y_1|$.

\item[Case 4.]$0 < k < |X_1|$.

Take a stable set $S$ in $G_2 - X_2$ greedily from left to right.  If $S$ hits $Y_2$, we remove $S$ from $G$, along with a colour class from $G_1$ intersecting $X_1$ but not $Y_1$.  Otherwise, we remove $S$ along with a colour class from $G_1$ intersecting both $X_1$ and $Y_1$.  In either case it is a simple matter to confirm that $\gamma_l^j(v)$ drops for every $v \in H_2$ as we did in Case 2.  We proceed by induction.

\item[Case 5.]$k=|Y_1|=|X_1|=1$.

In this case $|X_1|=k=1$.  If $G_2$ is not connected then $X_1$ and $Y_1$ are both clique cutsets and we can proceed as in Case 1.  If $G_2$ is connected and contains an $l$-clique, then there is some $v \in V_2$ of degree at least $l$ in the $l$-clique.  Thus $\gamma_l^j(H_2) > l$, contradicting our assumption that $l \geq \gamma_l^j(H_2)$.  So $\omega(G_2)<l$.  We can $\omega(G_2)$-colour $G_2$ in linear time using only colours not appearing in $X_1 \cup Y_1$, thus extending the $l$-colouring of $G_1$ to a proper $l$-colouring of $G$.

\item[Case 6.]$k=|Y_1|=|X_1|> 1$.

Suppose that $k$ is not minimal.  That is, suppose there is a vertex $v \in X_1 \cup Y_1$ whose closed neighbourhood does not contain all $l$ colours in the colouring of $G_1$.  Then we can change the colour of $v$ and apply Case 4.  So assume $k$ is minimal.

Therefore every vertex in $X_1$ has degree at least $l+|X_2|-1$.  Since $X_1\cup X_2$ is a clique, $\gamma_l^j(H_2) \geq l \geq \frac 12 (l+|X_2|+|X_1|+|X_2|)$, so $2|X_2|\leq l-k$.  Similarly, $2|Y_2|\leq l-k$, so $|X_2|+|Y_2|\leq l-k$.  Since there are $l-k$ colours not appearing in $X_1\cup Y_1$, we can $\omega(G_2)$-colour $G_2$, then permute the colour classes so that no colour appears in both $X_1\cup Y_1$ and $X_2 \cup Y_2$.  Thus we can extend the $l$-colouring of $G_1$ to an $l$-colouring of $G$.
\end{itemize}

These cases cover every possibility, so we need only prove that the colouring can be found in $O(nm)$ time.  If $k$ has been maximized and we apply induction, $k$ will stay maximized:  every vertex in $X_1 \cup Y_1$ will have every remaining colour in its closed neighbourhood except possibly if we recolour a vertex in Case 6.  In this case the overlap in what remains is $k-1$, which is the most possible since we remove a vertex from $X_1$ or $Y_1$, each of which has size $k$.  Hence we only need to maximize $k$ once.  We can determine which case applies in $O(m)$ time, and it is not hard to confirm that whenever we extend the colouring in one step our work can be done in $O(nm)$ time.  When we apply induction, i.e.\ in Cases 2, 4, and possibly 6, all our work can be done in $O(m)$ time.  Since $l<n$ it follows that the entire $l$-colouring can be completed in $O(nm)$ time.
\end{proof}

\subsection{Putting the pieces together}

We can now prove an algorithmic version of Theorem \ref{thm:ql}.

\begin{theorem}
Let $G$ be a quasi-line graph on $n$ vertices and $m$ edges.  Then we can find a proper colouring of $G$ using $\gamma(G)$ colours in $O(n^3m^2)$ time.
\end{theorem}

\begin{proof}We proceed by induction on $n$.  As already explained, we need only consider graphs containing no clique cutsets since $n^3m^2 \geq nm$.  We begin by applying Lemma \ref{lem:homo} at most $m$ times in order to find a quasi-line subgraph $G'$ of $G$ such that $\chi(G)=\chi(G')$, and given a $k$-colouring of $G'$, we can find a $k$-colouring of $G$ in $O(n^2m^2)$ time.  We must now colour $G'$.

If $G'$ is a circular interval graph we can determine this and $\gamma_l(G)$-colour it in $O(n^{3/2})$ time.  If $G'$ is a line graph we can determine this in $O(m)$ time using an algorithm of Roussopoulos \cite{roussopoulos73}, then $\gamma_l(G)$-colour it in $O(n^{2})$ time.  Otherwise, $G'$ must admit a canonical interval 2-join.  In this case Lemma 6.18 in \cite{kingthesis}, due to King and Reed, tells us that we can find such a decomposition in $O(n^2m)$ time.

This canonical interval 2-join $((V_1,X_1, Y_1),(V_2,X_2, Y_2))$ leaves us to colour the induced subgraph $G_1$ of $G'$, which has at most $n-1$ vertices and is quasi-line.  Given a $\gamma_l(G)$-colouring of $G_1$ we can $\gamma_l(G)$-colour $G'$ in $O(nm)$ time, then reconstruct the $\gamma_l(G)$-colouring of $G$ in $O(n^2m^2)$ time.  The induction step takes $O(n^2m^2)$ time and reduces the number of vertices, so the total running time of the algorithm is $O(n^3m^2)$.
\end{proof}

\noindent{\bf Remark: }With some care and using more sophisticated results on decomposing quasi-line graphs (see \cite{chudnovskyk11}), we believe it should be possible to reduce the running time of the entire algorithm to $O(m^2)$.

\bibliography{masterbib}
\end{document}